\documentclass[letterpaper,11pt]{article}
\input epsf.tex
\usepackage[dvips]{graphicx}
\usepackage{epstopdf}
\usepackage{graphicx}
\usepackage{amsthm,amsmath,amssymb}

\newtheorem{lem}{Lemma}
\newtheorem{thm}{Theorem}
\newtheorem{cor}{Corollary}
\newcommand{\ba}{\begin{eqnarray}}
\newcommand{\ea}{\end{eqnarray}}
\newcommand{\ban}{\begin{eqnarray*}}
\newcommand{\ean}{\end{eqnarray*}}
\newcommand{\braket}[2]{\mbox{$ \langle #1 | #2 \rangle $}}

\newcommand{\ket}[1]{\mbox{$ | #1 \rangle $}}

\DeclareMathOperator{\supp}{Supp}
\DeclareMathOperator{\nz}{Nonzero}

\begin{document}

\title{Psi-Epistemic Theories: The Role of Symmetry}
\author{ Scott Aaronson\thanks{%
MIT. \ email: aaronson@csail.mit.edu. \ This material is based upon work
supported by the National Science Foundation under Grant No.\ 0844626. \ Also
supported by an NSF STC grant, a TIBCO Chair, a Sloan
Fellowship, and an Alan T.\ Waterman Award.},  Adam Bouland\thanks{%
MIT. \ email: adam@csail.mit.edu. \ This material is based upon work supported
by the National Science Foundation Graduate Research Fellowship under Grant
No.\ 1122374. \ Also supported by the Center for Science of Information
(CSoI), an NSF Science and Technology Center, under grant agreement
CCF-0939370.},  Lynn Chua\thanks{%
MIT. \ email: chualynn@mit.edu.},  George Lowther\thanks{%
email: george.lowther@blueyonder.co.uk.}}
\date{}
\maketitle

\begin{abstract}
Formalizing an old desire of Einstein, ``$\psi$-epistemic theories'' try to
reproduce the predictions of quantum mechanics, while viewing quantum states
as ordinary probability distributions over underlying objects called ``ontic
states.'' \ Regardless of one's philosophical views about such theories, the
question arises of whether one can cleanly rule them out, by proving no-go
theorems analogous to the Bell Inequality. \ In the 1960s, Kochen and
Specker (who first studied these theories) constructed an elegant $\psi$%
-epistemic theory for Hilbert space dimension $d=2$, but also showed that
any deterministic $\psi$-epistemic theory must be ``measurement contextual''
in dimensions $3$ and higher. \ Last year, the topic attracted renewed
attention, when Pusey, Barrett, and Rudolph (PBR) showed that any $\psi$%
-epistemic theory must ``behave badly under tensor product.'' \ In this
paper, we prove that even without the Kochen-Specker or PBR assumptions,
there are no $\psi$-epistemic theories in dimensions $d\geq 3$ that satisfy
two reasonable conditions: (1) symmetry under unitary transformations, and
(2) ``maximum nontriviality'' (meaning that the probability distributions
corresponding to any two non-orthogonal states overlap). \ This no-go theorem
holds if the ontic space is either the set of quantum states or the set of unitaries. \ The proof of this result, in
the general case, uses some measure theory and differential geometry. \ On the other
hand, we also show the surprising result that \emph{without} the symmetry
restriction, one can construct maximally-nontrivial $\psi$-epistemic
theories in every finite dimension $d$.
\end{abstract}

\section{Introduction}

Debate has raged for almost a century about the interpretation of the
quantum state. \ Although a quantum state evolves in a unitary and
deterministic manner according to the Schr\"{o}dinger equation, measurement
is a probabilistic process in which the state is postulated to collapse to a
single eigenstate. \ This is often viewed as an unnatural and
poorly-understood process.

$\psi$-epistemic theories have been proposed as alternatives to standard quantum
mechanics. \ In these theories, a quantum state merely represents
probabilistic information about a ``real, underlying'' physical state
(called the \emph{ontic state}). \ Perhaps not surprisingly, several no-go theorems
have been proven that strongly constrain the ability of $\psi$-epistemic
theories to reproduce the predictions of standard quantum mechanics. \ Most
famously, the Bell inequality \cite{bell}---while not usually seen as a result about $\psi$-epistemic
theories---showed that no such theory can account for the results of all possible measurements on an entangled state in a
``factorizable'' way (i.e., so that the ontic state has a separate component for each qubit, and measurements of a given qubit only reveal information about that qubit's component of the ontic state). \ Also, the Kochen-Specker theorem \cite{ks} showed that in Hilbert space
dimensions $d\geq 3$, no $\psi$-epistemic theory can be both deterministic and
``noncontextual'' (meaning that whether an eigenstate $\psi$ gets returned
as a measurement outcome is independent of which \emph{other} states are
also in the measurement basis). \ More recently, the Pusey-Barrett-Rudolph
(PBR) theorem \cite{pbr} showed that nontrivial $\psi$-epistemic theories are inconsistent, if
the ontic distribution for a product state $| \psi \rangle \otimes %
 | \phi \rangle $ is simply the tensor product of the ontic
distribution for $ | \psi \rangle $ with the ontic distribution for
$ | \phi \rangle $. \ Even more recently, papers by Maroney \cite{mar} and Leifer and Maroney \cite{marleif} prove the impossibility of a ``maximally $\psi$-epistemic theory," in which the overlap of the ontic distributions for all non-orthogonal states fully accounts for the uncertainty in distinguishing them via measurements. \

In this paper, we study what happens
if one drops the Bell, Kochen-Specker, and PBR assumptions, and merely asks for a $\psi$-epistemic theory in which the ontic distributions overlap for all non-orthogonal states. \

A $\psi$-epistemic theory is a particular type of ontological theory of quantum mechanics. Formally, an ontological theory in $d$ dimensions specifies:

\begin{enumerate}
\item A measurable space $\Lambda$, called the \emph{ontic space} (the
elements $\lambda \in \Lambda$ are then the ontic states).

\item A function mapping each quantum state $ | \psi \rangle \in H_d
$ to a probability measure $\mu_\psi$ over $\Lambda$, where $H_d$ is the
Hilbert space in $d$ dimensions.

\item For each orthonormal measurement basis $M=\{\phi_1,\ldots,\phi_d\}$, a
set of $d$ \emph{response functions} $\{\xi_{k,M}(\lambda) \in [0,1]\}$,
which give the probability that an ontic state $\lambda$ would produce a
measurement outcome $\phi_k$.
\end{enumerate}

The response functions
must satisfy the following two conditions:
\ba \int_\Lambda \xi_{k,M}(\lambda)\,\mu_\psi(\lambda)\,d\lambda &=& |\braket{\phi_k}{\psi}|^2\,, \label{born} \\
\sum_{i=1}^{d} \xi_{i,M}(\lambda) &=& 1\,\,\, \forall \lambda,\, M. \label{cover} \ea
Here Equation (\ref{born}) says that the ontological theory perfectly reproduces the predictions of quantum
mechanics (i.e., the Born rule). \ Meanwhile, Equation (\ref{cover}) says that the probabilities of the
possible measurement outcomes must always sum to $1$, even when ontic states
are considered individually (rather than as elements of probability
distributions). \ Note that Equations (\ref{born}) and (\ref{cover}) are
logically independent of each other.\footnote{Also, we call an ontological theory \textit{deterministic} if the response functions take values
only in $\{0,1\}$. \ The Kochen-Specker theorem then states that, in
dimensions $d \geq 3$, any deterministic theory must have response functions
that depend nontrivially on $M$.}

The conditions above can easily be satisfied by setting $\Lambda = \mathbb{CP}^{d-1}$, the complex projective space consisting of unit vectors in $H_d$ up to an arbitrary phase, and $%
\mu_\psi(\lambda) = \delta(\lambda-\psi)$, where $\delta$ is the Dirac delta
function, and $\xi_{k,M}(\lambda) = | \langle \phi_k | \lambda
\rangle |^2$. \ But that simply gives an uninteresting restatement of
quantum mechanics, since the $\mu_{\psi}$'s for different $\psi$'s have disjoint supports. \ An ontological
theory in which the $\mu_{\psi}$'s have disjoint supports is known as a \textit{$\psi$-ontic} theory \cite{harriganspekkens}. \ Let $\supp(\mu_{\psi})\subseteq\Lambda$ be the
support of $\mu_{\psi}$. \ Then we call an ontological theory \textit{$\psi$-epistemic} if there exist $\psi\neq\phi$ such that $\mu_\psi$ and $\mu_\phi$
have total variation distance less than 1 \cite{harriganspekkens}, i.e.
\begin{eqnarray}
\frac{1}{2}\int_{\Lambda}\left|\mu_{\psi}(\lambda)-\mu_{\phi}(\lambda)%
\right| d\lambda < 1.
\end{eqnarray}
If $\mu_\psi$ and $\mu_\phi$ have total variation distance less than 1, then
we say that $ | \psi \rangle $ and $ | \phi \rangle $ have
``nontrivial overlap''. \ Otherwise we say they have ``trivial overlap''. \
Note that it's possible for $ | \psi \rangle $ and $ | \phi
\rangle $ to have trivial overlap even if $\mu_\psi$ and $\mu_\phi$ have
intersecting supports (this can happen if $\supp(\mu_{\psi})\cap \supp%
(\mu_{\phi})$ has measure $0$).

Note also that if $ | \psi \rangle $ and $ | \phi \rangle $
are orthogonal, then if we set $ | \phi_1 \rangle  =  | \psi
\rangle $ and $ | \phi_2 \rangle  = | \phi \rangle $,
the conditions $| \langle \phi_1 | \psi \rangle |=| \langle
\phi_2 | \phi \rangle |=1$ and $| \langle \phi_2 | \psi \rangle %
|=| \langle \phi_1 | \phi \rangle |=0$ imply that $\mu_\psi$ and $%
\mu_\phi$ have trivial overlap. \ Hence, we call a theory \textit{%
maximally nontrivial} if the overlap is \emph{only} trivial for orthogonal states: that is, if all non-orthogonal states $%
 | \psi \rangle, | \phi \rangle $ have nontrivial
overlap.

In a maximally nontrivial theory, some of the uncertainty of quantum
measurement is explained by the overlap between the distributions
corresponding to non-orthogonal states. \ Recently Maroney \cite{mar} and
Leifer and Maroney \cite{marleif} showed that it is impossible to have a ``maximally $\psi$%
-epistemic theory" in which \emph{all} of the uncertainty is explained by
the overlap of distributions. \ Specifically, they require that, for all quantum states $\ket{\psi}, \ket{\phi}$,
\ba \int_{\supp(\mu_{\phi})} \mu_{\psi}(\lambda) d\lambda = |\braket{\phi}{\psi}|^2.\ea
Here we are asking for a much weaker
condition, in which only \emph{some} of the uncertainty in measurement
statistics is explained by the overlap of distributions, and we do not impose any conditions on the amount of overlap.

Another property that we might like a $\psi$-epistemic theory to satisfy is
\textit{symmetry}. \ Namely, we call a $\psi$-epistemic theory \textit{%
symmetric} if $\Lambda=\mathbb{CP}^{d-1}$ and the probability distribution $\mu_\psi(\lambda)$ is
symmetric under unitary transformations that fix $ | \psi \rangle $%
---or equivalently, if $\mu_{\psi}$ is a function $f_{\psi}$ only of $| \langle
\psi | \lambda \rangle |$. \ We stress that this function is allowed to be different for different $\psi$'s: symmetry
only applies to each $\mu_{\psi}$ individually. \ This makes our no-go theorem for symmetric theories stronger. \ If
additionally $\mu_{\psi}$ is a \emph{fixed} function $f$ only of $| \langle
\psi | \lambda \rangle |$, then we call the theory \emph{strongly symmetric}. \ Note that, if a theory is strongly symmetric, then in order to apply a unitary $U$ to a state $|\psi \rangle$, one can simply apply $U$ to the ontic states. \ So strongly symmetric theories have a clear motivation: namely, they allow us to keep the Schr\"{o}dinger equation as the time evolution of our system.

A similar notion to symmetry was recently explored by Hardy \cite{hardy} and Patra et al.\ \cite{patra}. \ Given a $\psi$-epistemic theory, it is natural to consider the action of unitaries on the ontic states $\lambda \in \Lambda$. \ Hardy and Patra et al.\ define such a theory to obey ``ontic indifference" if for any unitary $U$ such that $U\ket{\psi}=\ket{\psi}$, and any $\lambda \in \supp(\mu_\psi)$, we have $U\lambda=\lambda$. \ They then show that no $\psi$-epistemic theories satisfying ontic indifference exist in dimensions $d\geq 2$. \ Note that symmetric theories and even strongly symmetric theories need not obey ontic indifference, since unitaries can act nontrivially on ontic states in $\supp(\mu_\psi)$. \ So the result of Hardy and Patra et al.\ is incomparable with ours.

In dimension $2$, there exists a
strongly symmetric and maximally nontrivial theory found by Kochen and Specker \cite%
{ks}. \ In dimensions $d\geq 3$, Lewis et al.\ \cite{lewis} found a
nontrivial $\psi$-epistemic theory for all finite $d$. \ However, their
theory is not symmetric and is far from being maximally nontrivial.

In this paper, we first give a construction of a maximally nontrivial $\psi$%
-epistemic theory for arbitrary $d$. \ Our theory builds on that of
Lewis et al.\ \cite{lewis}, and was first constructed in a post on
MathOverflow \cite{mathoverflow}. \ Unfortunately, this theory is rather unnatural and is
not symmetric. \ We then prove that it is impossible to construct a
maximally nontrivial theory that \emph{is} symmetric, for Hilbert space
dimensions $d\geq 3$. \ Furthermore, we extend this work to rule out a generalization of strongly symmetric theories with $\Lambda=U(d)$ rather than $\Lambda = \mathbb{CP}^{d-1}$ in $d\geq 3$. \ In short, if we want maximally nontrivial
theories in $3$ or more dimensions, then we either need an
ontic space $\Lambda$ other than $\Lambda = \mathbb{CP}^{d-1}$ or $\Lambda=U(d)$, or else we need ontic distributions $\mu_{\psi}$ that
``single out preferred directions in Hilbert space.''

\section{Nonsymmetric, Maximally Nontrivial Theory}

By considering $\Lambda = \mathbb{CP}^{d-1} \times [0,1]$, Lewis et al.\ \cite{lewis}
found a deterministic $\psi$-epistemic theory for all finite $d$%
. \ They raised as an open problem whether a \emph{maximally} nontrivial
theory exists. \ In this section, we answer their question in the
affirmative. \ Specifically, we first show
that, for any two non-orthogonal states, we can construct a theory such that
their probability distributions overlap. \ We then take a convex combination
of such theories to obtain a maximally nontrivial theory.

\begin{lem}
\label{mix} Given any two non-orthogonal quantum states $ | a
\rangle, | b \rangle $, there exists a $\psi$-epistemic
theory $T(a,b)=(\Lambda,\mu,\xi)$ such that $\mu_a$ and $%
\mu_b$ have nontrivial overlap. \ Moreover, for $T(a,b)$, there
exists $\varepsilon>0$ such that $\mu_{a^{\prime }}$ and $\mu_{b^{\prime
}}$ have nontrivial overlap for all $ | a' \rangle $, $
| b' \rangle $ that satisfy

\[ ||a-a^{\prime }||,||b-b^{\prime
}||<\varepsilon. \]
\end{lem}

\begin{proof}
Our ontic state space will be $\Lambda=\mathbb{CP}^{d-1}\times [0,1]$. \ Given an orthonormal basis $M=\{\phi_1,\ldots,\phi_d\}$, we first sort the $\phi_i$'s in decreasing order of $\min(|\braket{\phi_i}{a}|,|\braket{\phi_i}{b}|)$. \ Then the outcome of measurement $M$ on ontic state $(\lambda,p)$ will be the smallest positive integer $i$ such that
\ba |\braket{\phi_1}{\lambda}|^2+\cdots+|\braket{\phi_{i-1}}{\lambda}|^2\leq p \leq |\braket{\phi_1}{\lambda}|^2+\cdots+|\braket{\phi_i}{\lambda}|^2. \ea

\noindent In other words, $\xi_{i,M}(\ket{\lambda},p)=1$ if $i$ satisfies the above and no $j<i$ does, and is $0$ otherwise. \ If we assume that $\mu_{\psi}(\ket{\lambda},p)=\delta(\ket{\lambda}-\ket{\psi})$ for all $p\in[0,1]$, then it can be verified that $T(a,b)$ is a valid ontological theory, albeit so far a $\psi$-ontic one.

We now claim that there exists an $\varepsilon>0$ such that, for all orthonormal bases $M=\{\phi_1,\ldots,\phi_d\}$, there exists an $i$ such that $|\braket{\phi_i}{a}|\geq\varepsilon$ and $|\braket{\phi_i}{b}|\geq\varepsilon$. \ Indeed, by the triangle inequality, we can let $\varepsilon=|\braket{a}{b}|/d$, and $\varepsilon >0$ since $|\braket{a}{b}|>0$. \ This means that, for all measurements $M$ and all $p\in[0,\varepsilon]$, the outcome is always $i=1$ when $M$ is applied to either of the ontic states $(\ket{a},p)$ or $(\ket{b},p)$.

Following Lewis et al.\ \cite{lewis}, we can ``mix" the probability distributions $\mu_a$ and $\mu_b$, or have them intersect in the region $p\in [0,\varepsilon]$, without affecting the Born rule statistics for any measurement. \ Explicitly, we can let
\ba E_{a,b} = \{\ket{a},\ket{b}\} \times [0,\varepsilon], \ea
so that all $\lambda\in E_{a,b}$ give the same measurement outcome $\phi_1$ for all measurements $M$. \ Then any probability assigned by $\mu_{a}$ or $\mu_{b}$ to states within $E_{a,b}$ can be redistributed over $E_{a,b}$ without changing the measurement statistics. \ Thus, we can define $\mu_{a}$ such that the weight it originally placed on $\ket{a} \times [0,\varepsilon]$ is now placed uniformly on $E_{a,b}$. \ More formally, we set
\ba \mu_{a}(\ket{\lambda},x) =
  \begin{cases}
   \delta(\ket{\lambda}-\ket{a})  & \text{if } x > \varepsilon \\
   \varepsilon\mu_{E_{a,b}}(\ket{\lambda},x)       & \text{if } x\leq \varepsilon,
  \end{cases} \ea
where $\mu_{E_{a,b}}$ is the uniform distribution over $E_{a,b}$. \ We similarly define $\mu_{b}$. \ This then yields a $\psi$-epistemic theory with nontrivial overlap between $\ket{a}$ and $\ket{b}$.

Furthermore, suppose we have $\ket{a'}$, $\ket{b'}$, such that $||a-a'||, ||b-b'||<\frac{\varepsilon}{2}$. \ Then by continuity, we can similarly mix the distributions $\mu_{a'}$ and $\mu_{b'}$, or have them intersect each other in the region $p\in[0,\frac{\varepsilon}{2}]$, without affecting any measurement outcome. \ Note that the procedure of sorting the basis vectors of $M$ might cause the measurement outcome to change discontinuously. \ However, this is not a problem since the procedure depends only on $\ket{a}$ and $\ket{b}$, which are fixed, and hence occurs uniformly for all $\ket{a'}$ and $\ket{b'}$ defined as above.
\end{proof}

Lemma \ref{mix} implies that for any two non-orthogonal states $ |
a \rangle $ and $ | b \rangle $, we can construct a theory
where $\mu_{a^{\prime }}$ and $\mu_{b^{\prime }}$ have nontrivial
overlap for all $||a-a^{\prime }||, ||b-b^{\prime }||<\varepsilon$,
for some $\varepsilon>0$. \ To obtain a maximally nontrivial theory, such that
any two non-orthogonal vectors have probability distributions that overlap,
we take a convex combination of such $\psi$-epistemic theories.

Given two $\psi$-epistemic theories $T_1=(\Lambda_1,\mu_1,\xi_1)$ and $%
T_2=(\Lambda_2,\mu_2,\xi_2)$ and a constant $c\in(0,1)$, we define the new
theory $cT_1+(1-c)T_2=(\Lambda_c,\mu_c,\xi_c)$ by setting $%
\Lambda_c=\left(\Lambda_1 \times \{1\} \right) \cup \left(\Lambda_2 \times
\{2\} \right)$ and $\mu_c = c\mu_1 + (1-c)\mu_2$. \ For any $%
(\lambda,i)\in\Lambda_c$, we then define $\xi_c$ to equal $\xi_i$ on $%
\Lambda_i$.

The following is immediate from the definitions.

\begin{lem}
\label{convex} $cT_1+(1-c)T_2$ is a $\psi$-epistemic theory. \ Furthermore,
if $T_1$ mixes the probability distributions $\mu_{\psi}$, $\mu_{\phi}$ of
two states $ | a \rangle $ and $ | b \rangle $, and $%
T_2$ mixes $\mu_{a^{\prime }}$ and $\mu_{b^{\prime }}$, then $%
cT_1+(1-c)T_2$ mixes both pairs of distributions, assuming $c\not\in\{0,1\}$.
\end{lem}

Note that the ontic state space of a convex combination of theories contains a copy of each of the original ontic spaces $\Lambda_1$ and $\Lambda_2$. If $\Lambda_1 =\Lambda_2$, it is natural to ask if we could get away with keeping only one copy of the ontic state space. \ Unfortunately the answer in general is no. \ Suppose that we let $\Lambda_c= \Lambda_1 =\Lambda_2$, let $\mu_c = c\mu_1 + (1-c)\mu_2$, and let $\xi_c = c\xi_1 + (1-c)\xi_2$. \ Then the probability of measuring outcome $i$ under measurement $M$ and ontic distribution $\mu_{c_\psi}$ is
\[\int_\Lambda \left(c\xi_{1_{i,M}}(\lambda)+(1-c)\xi_{2_{i,M}}(\lambda)\right)\,\left(c\mu_{1_\psi}(\lambda)+(1-c)\mu_{2_\psi}(\lambda)\right)\,d\lambda   \]
which will not in general reproduce the Born rule due to unwanted cross terms. \ This is why it is necessary to keep two copies of the ontic state space when taking a convex combination of theories.

Using Lemmas \ref{mix} and \ref{convex}, we now construct a maximally
nontrivial $\psi$-epistemic theory. \ Let $T(a,b)$ be the theory
returned by Lemma \ref{mix} given $ | a \rangle , | b
\rangle \in H_d$. \ Also, for all positive integers $n$, let $A_n$ be a $%
1/n$-net for $H_d$, that is, a finite subset $A_n\subseteq H_d$ such that
for all $ | a \rangle \in H_d$, there exists $ | a'
\rangle \in A_n$ satisfying $||a-a^{\prime }||<1/n$. \ By making
small perturbations, we can ensure that $ \langle a | b \rangle
\neq 0$ for all $ | a \rangle , | b \rangle \in A_n$. \ Then our theory $T$ is defined as follows:
\begin{eqnarray}
T = \frac{6}{\pi^2}\sum_{n=1}^{\infty}\frac{1}{n^2}\left(\frac{1}{|A_n|^2}%
\sum_{a,b\in A_n} T(a,b)\right).
\end{eqnarray}

\noindent (Here, in place of $6/(\pi^{2} n^{2})$, we could have chosen any infinite sequence summing to unity.) \ This yields a maximally nontrivial theory, since it can be verified that $%
\mu_{a}$ and $\mu_{b}$ have nontrivial overlap for all non-orthogonal
states $ | a \rangle $ and $ | b \rangle $. \ Note
that the ontic space is now $\mathbb{CP}^{d-1}\times [0,1]\times \mathbb{N}$, which has the same
cardinality as $\mathbb{CP}^{d-1}$. \ It is thus possible to map this theory into a theory
that uses $\Lambda=\mathbb{CP}^{d-1}$ as its ontic space, using a bijection between the
ontic spaces. \ However, it is clear that under such a bijection the theory
becomes less symmetric: the quantum state $ | a \rangle $ no
longer has any association with the state $ | a \rangle $ in the
ontic space, and the measure is also very unnatural.

\section{Nonexistence of Symmetric, Maximally Nontrivial Theories}

\label{nonexistence} We now turn to showing that it is impossible to
construct a \textit{symmetric} maximally nontrivial theory, in dimensions $%
d\geq 3$. \ Recall that a theory is called symmetric if

\begin{enumerate}
\item $\Lambda=\mathbb{CP}^{d-1}$, and

\item for any quantum state $ | \psi \rangle $, the associated
probability distribution $\mu_{\psi}$ is invariant under unitary
transformations that preserve $ | \psi \rangle $.
\end{enumerate}

Specifically, if $U$ is a unitary transformation such that $U | \psi
\rangle = | \psi \rangle $, then we require that $%
\mu_{U\psi}(U\lambda)=\mu_{\psi}(\lambda)$. \ This implies that $%
\mu_{\psi}(\lambda)$ is a function only of $| \langle \psi | \lambda
\rangle |^2$: that is,
\begin{eqnarray}
\mu_{\psi}(\lambda)=f_{\psi}\left(| \langle \psi | \lambda \rangle %
|^2\right)
\end{eqnarray}
for some nonnegative function $f_{\psi}$. \ In other words, the probability
measure $\mu_\psi$ associated with state $\psi$ must be a measure $\nu_\psi$
on the unit interval which has been ``stretched out" onto $H_d$ over curves
of constant $|\langle \psi | \lambda \rangle |^2$. \ If additionally we assume that for \emph{any} $U$, $\mu_{U\psi}(U\lambda)=\mu_{\psi}(\lambda)$, or equivalently that $\mu_{\psi}(\lambda)=f\left(| \langle \psi | \lambda \rangle %
|^2\right)
$
for some \emph{fixed} nonnegative function $f$, the theory is called strongly symmetric.

In this section, we first prove several facts about symmetric, maximally nontrivial theories in general. \ Using these facts, we then show that no \emph{strongly} symmetric, maximally nontrivial theory exists in dimension 3 or higher. \ Restricting to the strongly symmetric case will make the proof considerably easier. \ Later we will show how to generalize to the ``merely" symmetric case.

\begin{figure}[h] \centering
\includegraphics[height=6cm]{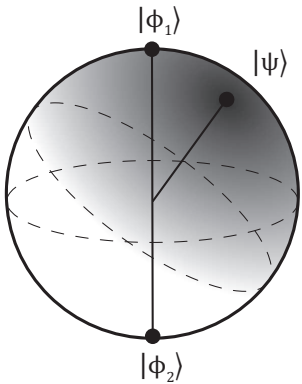} \vspace{-20pt}
\caption{Diagram of maximally nontrivial theory in $d=2$ on the Bloch sphere. The shaded region corresponds to $\supp(\mu_{\psi})$.} \label{figks}
\end{figure}

As mentioned earlier, Kochen and Specker proved that a strongly symmetric,
maximally nontrivial $\psi $-epistemic theory exists in dimension $d=2$ \cite%
{ks}. \ In their theory, which is illustrated in Figure \ref{figks}, the ontic space is $\Lambda =\mathbb{CP}^{1}$, and the
response functions for a given basis $M=\{\phi _{1},\phi _{2}\}$ are
\begin{eqnarray}
\xi _{1,M}(\lambda ) &=&1\text{ if }\left\vert  \langle \lambda |
\phi_1 \rangle \right\vert \geq \left\vert  \langle \lambda |
\phi_2 \rangle \right\vert \text{ or }0\text{ otherwise}, \\
\xi _{2,M}(\lambda ) &=&1\text{ if }\left\vert  \langle \lambda |
\phi_2 \rangle \right\vert \geq \left\vert  \langle \lambda |
\phi_1 \rangle \right\vert \text{ or }0\text{ otherwise}.
\end{eqnarray}%
Hence the response functions are deterministic and partition the ontic
space. \ Intuitively, the result of a measurement on any ontic state is
the state in the measurement basis to which it is closest. \ For any quantum
state $ | \psi \rangle \in H_{2}$, the probability distribution
over $\Lambda $ is given by
\begin{equation*}
\mu _{\psi }(\lambda )=%
\begin{cases}
\frac{2}{\pi }\left( | \langle \lambda | \psi \rangle |^{2}-\frac{1%
}{2}\right)  & \text{if }| \langle \lambda | \psi \rangle |^{2}>%
\frac{1}{2} \\
0 & \text{otherwise}.%
\end{cases}%
\end{equation*}%
It can readily be verified that this theory satisfies the conditions for a $\psi$-epistemic theory, and has the properties of being
strongly symmetric and maximally nontrivial. \ It is also maximally $\psi$-epistemic in the sense described by Maroney \cite{mar} and Maroney and Leifer \cite{marleif}. \

Given a measurement outcome $\psi$ and a basis $M$ containing $\psi$, we define the \emph{nonzero set} $\nz(\xi_{\psi,M})$ to be the set of ontic states $\lambda$ such that the response function $\xi_{\psi,M}(\lambda)$ gives a nonzero probability of returning $\psi$ when $M$ is applied:
\begin{eqnarray}
\nz(\xi_{\psi,M}) = \left\{\lambda: \xi_{\psi,M}(\lambda)\neq0\right\}.
\end{eqnarray}
Clearly in any $\psi$-epistemic theory, $\supp(\mu_\psi)
\subseteq \nz(\xi_{\psi,M})$ for any measurement basis $M$ that contains $\psi$, because the state $ | \psi \rangle $ must return measurement
outcome $\psi$ with probability 1 for any such $M$. \ Harrigan and Rudolph \cite{hr} call a $\psi$-epistemic theory \textit{deficient} if there exists a quantum state $ | \psi
\rangle $ and measurement basis $M$ containing $\psi$ such that
\begin{eqnarray}
\supp(\mu_{\psi}) \subsetneq \nz(\xi_{\psi,M}).  \label{deficiency}
\end{eqnarray}
In other words, a theory is deficient if there exists an ontic state $\lambda$ such that $\lambda$
has a nonzero probability of giving the measurement outcome corresponding to $%
 | \psi \rangle $ for some $M$, even though $\lambda\notin \supp(\mu_{\psi})
$. \ This can be thought of as a ``one-sided friendship" between $%
 | \psi \rangle $ and $\lambda$.

It was first pointed out by Rudolph \cite{rudolph}, and later shown by
Harrigan and Rudolph \cite{hr}, that theories in dimension $d\geq 3$ must be
deficient. \ In this section, we prove that as a result of deficiency, it is impossible
to have a symmetric, maximally nontrivial theory with $d\geq 3$. \ We derive a
contradiction by showing that if the theory is maximally nontrivial, then
there exist orthogonal states $ | \psi \rangle, |
\phi \rangle $, and a measurement basis $M$ containing $\ket{\psi}$, such that if $ | \phi
\rangle $ is measured, then the outcome $ | \psi
\rangle $ is returned with nonzero probability, contradicting the
laws of quantum mechanics.

We start with a few preliminary results on symmetric, maximally nontrivial theories. \
As stated previously, we know that $\mu_\psi$ is generated by stretching a
probability measure $\nu_\psi$ on the unit interval over $H_d$ along spheres
of constant $| \langle \psi | \lambda \rangle |^2$. \ By the Lebesgue
decomposition theorem, $\nu_\psi$ can be written uniquely as a sum of two
measures $\nu_{\psi, C}$ and $\nu_{\psi, S}$, where $\nu_{\psi ,C}$ is
absolutely continuous with respect to the Lebesgue measure over the unit
interval, and $\nu_{\psi, S}$ is singular with respect to that measure. \ Here when we say $\nu_{\psi, C}$ is ``absolutely continuous" with respect to the Lebesgue measure, we mean that it assigns zero measure to any set of Lebesgue measure zero. \ When we say $\nu_{\psi, S}$ is ``singular," we mean that its support is confined to a set of Lebesgue measure zero. \ Similarly, $\mu_{\psi}$ can be decomposed into its absolutely continuous and singular parts $\mu_{\psi,C}$ and $\mu_{\psi,S}$, which are defined respectively from the components $\nu_{\psi,C}$ and $\nu_{\psi,S}$ of $\nu_{\psi}$. \ By the Radon-Nikodym theorem, due to its absolute continuity $\nu_{\psi,C}$ has a probability
density function $g_{\psi}(x)$ that is a function, not a
pseudo-function or delta function. \ To simplify our analysis, first we will show that it is only necessary to look at the absolutely
continuous part of the distribution.

\begin{lem}
\label{measure_pos} For any distinct and non-orthogonal states $ | \psi \rangle $, $ | \phi \rangle $ in a
symmetric, maximally nontrivial theory, $\mu_{\psi, C}$ and $\mu_{\phi,C}$ have nontrivial overlap.
\end{lem}

\begin{proof}
Let $S_a$ denote the set of states $\lambda\in\Lambda$ with $| \langle \lambda | \psi \rangle |^2 = a$. \ If $a=1$, then $S_a$ is a single point with zero $\mu_{\phi}$ measure. \ For $0 < a < 1$, $S_a$ is a $(2d-3)$-sphere centered about $\psi$, and for $a=0$ it is a $(2d-4)$-dimensional manifold diffeomorphic to $\mathbb{CP}^{d-1}$. \ In both of the latter cases, as $\phi,\psi$ are distinct non-orthogonal states, the distribution of $| \langle \lambda | \phi \rangle |^2$ for $\lambda$ chosen uniformly on $S_a$ is absolutely continuous with respect to the Lebesgue measure on $[0,1]$. \ Therefore, the distribution of $| \langle \lambda | \phi \rangle |^2$ for $\lambda\in\Lambda$ chosen according to $\mu_\psi$ is absolutely continuous over $| \langle \lambda | \psi \rangle | < 1$.

By our symmetry condition, $\mu_{\phi,S}$ is the product of a singular measure on [0,1], denoted $\nu_{\phi,S}$, and the uniform measure on rings of constant $|\braket{\phi}{\lambda}|^2$. \ Since drawing $\lambda$ from $\mu_\psi$ induces an absolutely continuous measure on $|\braket{\phi}{\lambda}|^2$, then in particular $\mu_\psi$ has probability zero of producing a state $\lambda$ with $|\braket{\phi}{\lambda}|^2 \in \supp(\nu_{\phi,S})$, because $\supp(\nu_{\phi,S})$ is a set of measure zero. \ This implies that $\mu_\psi$ has probability zero of producing a state $\lambda \in \supp(\mu_{\phi,S})$. \ Hence there is zero overlap between $\mu_\psi$ and $\mu_{\phi,S}$. \ In particular, $\mu_{\psi,C}$ and $\mu_{\psi,S}$ have zero overlap with $\mu_{\phi,S}$. \ Similarly, $\mu_{\phi,C}$ and $\mu_{\phi,S}$ have zero overlap with $\mu_{\psi,S}$.

This shows that the overlap between $\mu_{\phi,C}$ and $\mu_{\psi,C}$ equals that between $\mu_{\phi}$ and $\mu_{\psi}$, which is nonzero for maximally nontrivial theories.
\end{proof}

From now on, we will assume $\mu_\psi$ is generated only from the absolutely
continuous part $\nu_{\psi,C}$, so that $\mu_\psi$ has as probability
density function $f_{\psi}\left(| \langle \psi | \lambda \rangle %
|^2\right)$ where $f_\psi$ is a function, not a pseudo-function. \ We can do
this without loss of generality, as our proof will not depend on the
normalization of the probability distributions, and will only use facts
about the absolutely continuous components of the measures.

Next, let the distance between two states $\psi$ and $\phi$ be defined by their scaled radial distance (also called the \emph{Fubini-Study metric}):
\ban ||\psi-\phi|| = \frac{2}{\pi}\arccos\left(|\braket{\psi}{\phi}|\right). \ean
For any state $ | \psi \rangle \in H_d$, with probability
distribution $\mu_{\psi}(\lambda)=f_{\psi}(| \langle \psi | \lambda
\rangle |^2)$, we define the \textit{radius} of $\mu_\psi$ to be the
distance between $ | \psi \rangle $ and the furthest away state at
which $\mu_\psi$ has substantial density:
\begin{eqnarray}
r_{\psi} = \mbox{sup}\left\{r: \forall \delta>0 \int_{\lambda: r-\delta
<||\psi-\lambda|| <r} \mu_\psi (\lambda) \,d \lambda >0 \right\}.
\label{defr}
\end{eqnarray}

\begin{lem}
\label{rsupp} For a symmetric theory, given any two states $ | \psi
\rangle, | \phi \rangle $, we have $||\psi-\phi||\geq
r_{\psi}+r_{\phi}$ if and only if $ | \psi \rangle $ and $ |
\phi \rangle $ have trivial overlap.
\end{lem}

\begin{proof}
Suppose that $||\psi-\phi||\geq r_{\psi}+r_{\phi}$ but $\ket{\psi}$ and $\ket{\phi}$ have nontrivial overlap. \ Then $\supp(\mu_{\psi})\cap \supp(\mu_{\phi})$ has nonzero measure, and for any $\lambda$ in that set, the triangle inequality implies that $r_{\psi}+r_{\phi}\geq ||\psi-\lambda||+||\phi-\lambda|| \geq ||\psi-\phi||$. \
Thus $r_{\psi}$ and $r_{\phi}$ satisfy $r_{\psi}+r_{\phi} = ||\psi-\phi||$, which is a contradiction since $||\psi-\lambda|| + ||\phi-\lambda|| = ||\psi-\phi||$ only on a set of measure zero. \

Now suppose that $||\psi-\phi||< r_{\psi}+r_{\phi}$. \ Consider $\lambda_{\mathrm{int}}$, an ontic state which lies at the intersection of rings of radii $r_{\psi}$ and $r_{\phi}$ about $\psi$ and $\phi$, respectively. \ In other words $||\psi-\lambda_{\mathrm{int}}|| =r_{\psi} $ and $||\phi-\lambda_{\mathrm{int}}|| =r_{\phi}$. \ Such a $\lambda_{\mathrm{int}}$ exists because $||\psi-\phi||< r_{\psi}+r_{\phi}$. \ Then in the neighborhood of $\lambda_{\mathrm{int}}$, we claim that $\mu_{\psi}$ and $\mu_{\phi}$ have nontrivial overlap.

To show this, we will define a set $B$ of positive measure, on which $\mu_{\psi}$ and $\mu_{\phi}$ are ``equivalent" to the Lebesgue measure, in the sense that if $S\subseteq B$ has positive Lebesgue measure, then $S$ has positive measure under both $\mu_{\psi}$ and $\mu_{\phi}$. \ This implies that $\psi$ and $\phi$ have nontrivial overlap on $B$.

By the symmetry condition, each $\mu_{\psi}$ is a product measure between a measure $\nu_{\psi}$ on $[0,1]$ and a uniform measure on surfaces of constant $|\braket{\psi}{\lambda}|$. \ Let $u$ and $v$ be the normal vectors to surfaces of constant $|\braket{\psi}{\lambda}|$ and $|\braket{\phi}{\lambda}|$ at $\ket{\lambda_{\mathrm{int}}}$, respectively. \ Note that $u$ ($v$) is equal to the tangent vector to the geodesic running from $\psi$ ($\phi$) to $\lambda_{\mathrm{int}}$ evaluated at $\lambda_{\mathrm{int}}$. \ Since $||\psi-\phi||< r_{\psi}+r_{\phi}$, these are distinct geodesics, so $u$ and $v$ are linearly independent.

Since $\mathbb{CP}^{d-1}$ is a smooth Riemannian manifold, $u$ and $v$ form a local coordinate system in the $\varepsilon$ neighborhood of $\ket{\alpha}$, which we denote $N_{\varepsilon}(\ket{\alpha})$. \ If we associate coordinates $x_1$, $x_2$ with $u$ and $v$, the integral over $N_{\varepsilon}(\ket{\alpha})$ can be parameterized as
\[
\int  g(x_1, x_2, y_1 \ldots y_{2d-4}) \mathrm{d} x_1 \mathrm{d}x_2 \mathrm{d}y_1 \ldots \mathrm{d}y_{2d-4}
\]
Here $g$ is the square root of the metric, which is strictly positive in the neighborhood of $\ket{\lambda_{\mathrm{int}}}$. \ Also, $\mathrm{d}x_i$ is the Lebesgue integral over the coordinate $x_i$, and the $y_i$ are coordinates corresponding to the remaining $2d-4$ dimensions of the space.

Now consider the set
\[
B= N_{\varepsilon}(\ket{\alpha}) \cap  \supp(\mu_{\psi}) \cap \supp(\mu_{\phi})
\]
Trivially $\mu_{\psi}$ and $\mu_\phi$ are equivalent to the Lebesgue measure on $B$. \ Note that $\supp(\mu_{\psi})$ is a union of surfaces $S_1$ of constant $| \braket{\lambda }{\psi}|$ which are perpendicular to $u$ at $\alpha$. \ If $\varepsilon$ is sufficiently small these surfaces have negligible curvature, so they look like orthogonal hyperplanes in the $x_1$ coordinate system.  \ Let $\varepsilon_1$ be the Lebesgue measure on $\supp(\mu_{\psi}) \bigcap B$. \ Let $S_2$ and $\varepsilon_2$ be defined similarly for $\phi$. \ If the surfaces $S_1$, $S_2$ had zero curvature, the Lebesgue measure of $B$ would simply be the product of the measures $\varepsilon_1 \varepsilon_2$, since $x_1$ and $x_2$ are orthogonal coordinates. \ Since the surfaces have slight curvature, and the coordinates $x_i$ are not truly orthogonal, the above calculation has to be changed slightly. \ Specifically, for sufficiently small $\varepsilon$ the Lebesgue measure of $B$ can be approximated by $g\varepsilon_1 \varepsilon_2$, where $g$ is the square root of the metric at $\lambda_{\mathrm{int}}$. \ This quantity is strictly positive since each $\varepsilon_i>0$ by the definition of $r$, the metric $g$ is strictly positive, and $\mu_{\phi}$ and $\mu_\phi$ are absolutely continuous with respect to the Lebesgue measure. \ Hence $B$ has positive Lesbesgue measure.

\end{proof}

\begin{cor}
\label{rorth} If $ | \psi \rangle $ and $ | \phi \rangle $
are orthogonal, then $r_{\psi}+r_{\phi}\leq 1$.
\end{cor}

\begin{proof}
If $\braket{\psi}{\phi}=0$, then $||\psi-\phi||=1$. \ Since any orthogonal $\ket{\psi}$ and $\ket{\phi}$ have trivial overlap, Lemma \ref{rsupp} implies that $r_{\psi}+r_{\phi}\leq 1$.
\end{proof}

\begin{lem}
\label{rhalf} Given any maximally nontrivial and symmetric theory in $d\geq 3
$, for any state $ | \psi \rangle \in H_d$, we have $r_{\psi}=\frac{%
1}{2}$.
\end{lem}

\begin{proof}
\begin{figure}[h] \centering
\includegraphics[height=5cm]{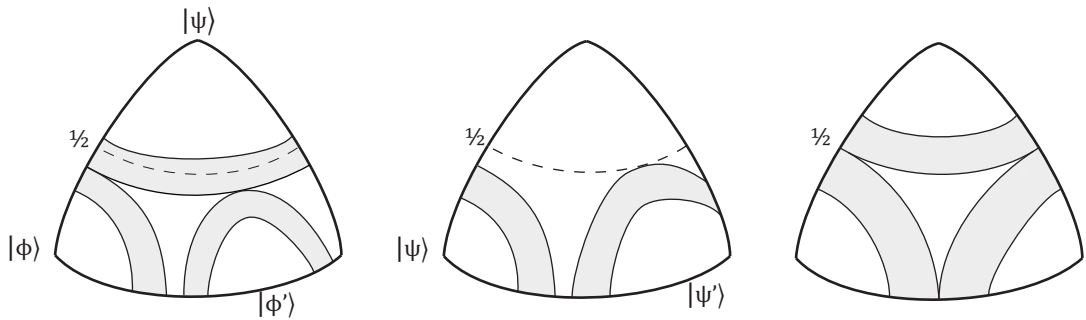}
\caption{From left to right: pictorial representations of the proof that $r_{\psi}\leq\frac{1}{2}$ in dimension $3$, the proof that $r_{\psi}=\frac{1}{2}$, and the form of the $\mu_{\psi}$'s that we ultimately deduce (with $r_{\psi}=\frac{1}{2}$ for all $\ket{\psi}\in H_d$). \ The shaded regions are the supports of the respective probability distributions.} \label{figlemma5}
\end{figure}
We first show that $r_{\psi}\leq \frac{1}{2}$ for all $\ket{\psi}\in H_d$, which we illustrate in the left side of Figure \ref{figlemma5} for the case where $d=3$. \ Suppose there exists $\ket{\psi}$ such that $r_{\psi}=\frac{1}{2}+\varepsilon$ for some $\varepsilon>0$. \ From Corollary \ref{rorth}, for all $\ket{\phi}$ orthogonal to $\ket{\psi}$, we have $r_{\phi}\leq \frac{1}{2}-\varepsilon$. \ In dimension $d\geq 3$, there exist non-orthogonal states $\ket{\phi}$, $\ket{\phi'}$ such that $\braket{\psi}{\phi}=\braket{\psi}{\phi'}=0$, and $\ket{\phi}\neq\ket{\phi'}$. \ Then $r_{\phi}+r_{\phi'}\leq 1-2\varepsilon$. \ If we choose $\ket{\phi}$, $\ket{\phi'}$ such that $1-2\varepsilon<||\phi-\phi'||<1$, then from Lemma \ref{rsupp}, we have that $\mu_{\phi}$ and $ \mu_{\phi'}$ have trivial overlap even though $\braket{\phi}{\phi'}\neq 0$. \ This contradicts the theory being maximally nontrivial.

We now show that $r_{\psi}\geq \frac{1}{2}$ for all $\ket{\psi}\in H_d$, as illustrated in the center of Figure \ref{figlemma5}. \ Suppose there exists $\ket{\psi}$ such that $r_{\psi}=\frac{1}{2}-\varepsilon$ for some $\varepsilon>0$. \ Since $r_{\psi'}\leq\frac{1}{2}$ for all $\ket{\psi'}\in H_d$, thus $r_{\psi}+r_{\psi'}\leq 1-\varepsilon$. \ If we choose $\ket{\psi}$, $\ket{\psi'}$ such that $1-\varepsilon < ||\psi-\psi'||<1$, then  $\mu_{\psi}$ and $ \mu_{\psi'}$ have trivial overlap from Lemma \ref{rsupp} even though $\braket{\psi}{\psi'}\neq 0$. \ This again contradicts maximum nontriviality.
\end{proof}

This immediately implies the following:

\begin{cor}
\label{maxnontriviallem} In dimensions $d\geq 3$, a symmetric $\psi$-epistemic theory is maximally
nontrivial if and only if for any state $ | \psi \rangle $ and for all $\delta >0$ the measure $%
\mu_\psi$ integrated over the following region is nonzero:
\begin{eqnarray}
\left\{\lambda:\frac{1}{2}\leq| \langle \psi | \lambda \rangle %
|^2\leq \frac{1}{2}+\delta\right\}  \label{maxnontrivial}
\end{eqnarray}
Moreover, $\supp(\mu_{\psi})$ has measure zero on the set of $\lambda$ such
that $| \langle \psi | \lambda \rangle |^2 <\frac{1}{2}$.
\end{cor}
\begin{proof}
By Lemma \ref{rhalf}, for any state $\ket{\psi}$ we have $r_{\psi}=\frac{1}{2}$. \ By rewriting the distance between states in terms of their inner product, the corollary follows from the definition of $r_{\psi}$ in Equation \ref{defr}.
\end{proof}

In Lemma \ref{rhalf}, we showed that the radius $r_\psi$ of every state $\psi$ in a maximally nontrivial symmetric theory is $\frac{1}{2}$. \ We now use this to show that a certain set of ontic states is deficient. \ Recall that we say a set $S$ is deficient for measurement $M$ if $S$ is not in $\supp(\mu_{\phi_i})$ for any $\phi_i \in M$.

\begin{cor}
\label{deficient} Given any symmetric, maximally nontrivial $\psi$-epistemic
theory in $d\geq 3$, for any measurement basis $M=\{\phi_i\}_{i=1}^d$, the
region \[R_M = \left\{\lambda:| \langle \phi_i | \lambda \rangle |^2 <%
\frac{1}{2}\,, i=1,\ldots,d\right\}\] is deficient except on a set of measure
zero. \ (Note that, by elementary geometry, $R_M$ has positive measure if and only if $d\geq 3$.)
\end{cor}

\begin{proof}
By Corollary \ref{maxnontriviallem}, for all $i=1,\ldots,d$, the set $\supp(\phi_i)$ must have measure zero over the region $R_M$. \ However, Equation \ref{cover} implies that any $\lambda\in R_M$ must be in $\nz(\xi_{i,M})$ for some $i$ even if it is not in $\supp(\phi_i)$. \ This means that $R_M$ is deficient except possibly on a set of measure zero.
\end{proof}

In general, deficiency occurs in any theory in $d\geq 3$ even without the
symmetry assumption, as proved by Harrigan and Rudolph \cite{hr} using the
Kochen-Specker theorem \cite{ks}. \ In Corollary \ref{deficient}, we showed
that symmetry implies a specific \emph{type} of deficiency.

To show that no strongly symmetric, maximally nontrivial theory exists, we first prove two simple
results for $\psi$-epistemic theories in general. \ These results will help us to derive a contradiction
for strongly symmetric, maximally nontrivial theories.

\begin{lem}
\label{orthogonal} Given any two orthogonal states $ | \phi \rangle $ and $%
 | \psi \rangle $, the set $\supp(\mu_{\phi})\cap \nz(\xi_{\psi,M})$ has
measure zero for all measurements $M$ that contain $\psi$.
\end{lem}

\begin{proof}
Suppose to the contrary that $\supp(\mu_{\phi})\cap \nz(\xi_{\psi,M})$ has positive measure for some measurement $M$ containing $\psi$. \ Then by definition, if the state $\ket{\phi}$ is measured using $M$, the outcome corresponding to $\ket{\psi}$ is returned with nonzero probability. \ But since $|\braket{\psi}{\phi}|^2=0$, this contradicts the Born rule (Equation (\ref{born})).
\end{proof}

\begin{lem}
\label{response} For any $\alpha\in\Lambda$, let $B_{\varepsilon}(\alpha)=\left%
\{\lambda:||\lambda-\alpha||<\varepsilon\right\}$ be an $\varepsilon$-ball around $
\alpha$, for some $\varepsilon>0$. \ Given a measurement basis $%
M=\{\phi_i\}_{i=1}^d$, there exists some $j$ such that
\begin{eqnarray}
\int_{B_\varepsilon(\alpha)} \xi_{j,M}(\lambda) \,d\lambda >0.
\end{eqnarray}
\end{lem}

\begin{proof}
For any such $\alpha\in\Lambda$ the following holds,
\ban \int_{B_\varepsilon(\alpha)} \sum_{i=1}^{d} \xi_{i,M}(\lambda)\,d\lambda =\int_{B_\varepsilon(\alpha)} 1\,d\lambda >0. \ean

\noindent This then implies that there exists some $j$ such that
\ban \int_{B_\varepsilon(\alpha)} \xi_{j,M}(\lambda) \,d\lambda >0. \ean
\end{proof}

\begin{figure}[h]
\centering \includegraphics[height=6cm]{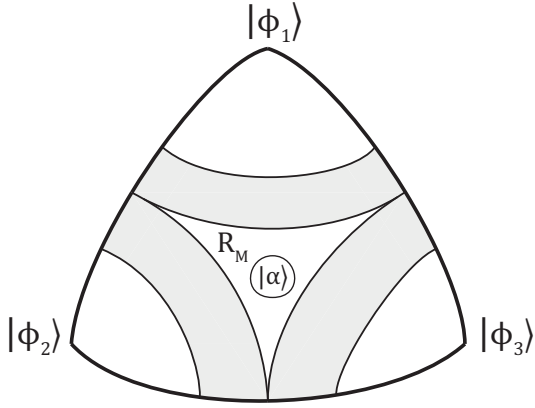} \vspace{-20pt}
\caption{Pictorial representation of the deficiency region for $d=3$. \ The
shaded regions are the supports of the respective probability measures,
and the middle unshaded region $R_M$ is deficient.} \label{figdeficiency}
\end{figure}

Using these two results, we can now prove that in dimension $d\geq 3$, there exists no strongly symmetric, maximally nontrivial $\psi$-epistemic theory.

\begin{thm}
\label{strong-symmetry-thm}
There exists no strongly symmetric, maximally nontrivial $\psi$-epistemic theory in
dimension $d\geq 3$.
\end{thm}

\begin{proof}
Suppose we have a symmetric, maximally nontrivial theory in dimension $d\geq 3$, and we fix a measurement basis $M=\{\phi_i\}_{i=1}^d$. \ From Corollary \ref{deficient}, there exists a deficiency region given by
\[ R_M = \left\{\lambda:|\braket{\phi_i}{\lambda}|^2 <\frac{1}{2}\,, i=1,\ldots,d\right\}, \]
\noindent perhaps minus a set of measure zero. \ This is illustrated in Figure \ref{figdeficiency} for the case where $d=3$.

Consider $\ket{\alpha}=\frac{1}{\sqrt{d}}\left(\ket{\phi_1}+\cdots +\ket{\phi_d}\right)$, which is contained in the deficiency region. \
Given $\varepsilon >0$, let $B_\varepsilon (\alpha)=\{\lambda:||\lambda-\alpha||<\varepsilon\}$ be the $\varepsilon$-ball around $\ket{\alpha}$. We choose $\varepsilon$ such that $B_\varepsilon (\alpha)$ is contained in $R_M$. \ From Lemma \ref{response}, there exists some $j$ such that $B := B_\varepsilon(\alpha)\cap\nz(\xi_{j,M})$ has nonzero measure. \  Without loss of generality, we assume that $j=1$.

Let $\nu$ be the measure obtained by averaging $\mu_{\psi}$ over all states $\ket{\psi}$ orthogonal to $\ket{\phi_1}$, and let $A$ be the set of all $\lambda$ such that $|\braket{\phi_1}{\lambda}|^2<\frac{1}{2}$. \ Since the theory is strongly symmetric, $\nu$ must be a function only of $|\braket{\phi_1}{\lambda}|^2$. \ Moreover, each of the measures $\mu_{\psi}$ assigns positive measure to the region of states $\lambda$ such that $|\braket{\psi}{\lambda}|^2$ is close to $\frac{1}{2}$, hence the averaged measure $\nu$ assigns positive measure to every open subset of $A$, and therefore in particular to $B$. \ This contradicts Lemma \ref{orthogonal}, which implies that each of the averaged measures $\mu_{\psi}$ must assign zero measure to $\nz(\xi_{1,M})$ and hence $B$.
\end{proof}

\subsection{Proof of Generalized No-Go Theorem}

We now generalize our proof of Theorem \ref{strong-symmetry-thm} to the ``merely" symmetric case, where the probability distributions $\mu_\psi$ can vary with $\psi$. \ First note that our previous proof does not immediately carry over. \ Since the probability distributions can vary as $\psi$ changes, it is possible that the distributions for states orthogonal to $\phi_1$ might be able to ``evade" the set $B$ in the proof of Theorem \ref{strong-symmetry-thm} which returns answer $\phi_1$ under measurement $M$, while maintaining some density near their outer radii.

To see how this might occur, consider the following one dimensional example: Let $\Lambda = \mathbb{R}$ be the real line. \ Construct $B\subseteq [0,1]$ to be a ``fat Cantor set" on $[0,1]$ as follows. \ Initially set $B=[0,1]$. \ In step 1, remove the middle $1/4$ of this interval, so that $B=[0,\frac{3}{8}]\cup[\frac{5}{8},1]$. \ At the $i^{\mathrm{th}}$ step, remove the middle $\frac{1}{2^{2i}}$ of each of the $2^i$ remaining intervals. \ Continue indefinitely. \ The resulting set $B$ is called a ``fat Cantor set" because it is nowhere dense (so contains no intervals), yet has positive Lebesgue measure on $[0,1]$.

For each point $x\in \mathbb{R}$, let $\mu_x$ be the uniform distribution on $[x-1,x+1]$ with $B$ removed. \ Then $\mu_x$ is absolutely continuous with respect to the Lebesgue measure for all $x\in \mathbb{R}$, and furthermore has positive measure on $[1+x-\varepsilon, 1+x]$ for all $\varepsilon>0$. \ However, despite the fact that $B$ has positive measure, the distributions $\mu_x$ never intersect $B$. \ The worry is that our distributions in $\mathbb{CP}^{d-1}$ could likewise evade the set $B$ in our proof, foiling our contradiction. \ This worry is related to a variant of the Kakeya/Besicovitch problem, as we discuss in Section \ref{sec-conclusion}.

We can extend Theorem \ref{strong-symmetry-thm} without solving a Kakeya-like problem, but to do so we will need a result about the differential geometry of $\mathbb{CP}^{d-1}$. \ Interestingly, we will use the fact that we are working in a complex Hilbert space; we believe the proof could be adapted to a real Hilbert space, but it would be much less convenient.

Discussing the differential geometry of $\mathbb{CP}^{d-1}$ is easiest if we first to pick a gauge for $\mathbb{CP}^{d-1}$, that is, if we pick a representative from each equivalence class of vectors which differ only by a global phase. \ We use the following gauge: let $\ket{\alpha}=\frac{1}{\sqrt{d}}\left(\ket{\phi_1}+\cdots +\ket{\phi_d}\right)$. \ For each equivalence class, we pick a representative $u$ such that $\braket{\alpha}{u}$ is real and positive. \ This uniquely identifies representatives for all equivalence classes of states, except those orthogonal to $\alpha$. \ Moreover, this way of choosing a gauge is continuous and smooth near $\alpha$; more precisely, equivalence classes which are close to one another have representatives which are also close to one another. \ This allows us to integrate over the manifold near $\alpha$ using these representatives. \ Using this gauge, we now prove the following.

\begin{lem}
\label{ui_lemma}
Let $M$ be a measurement basis $\{\phi_i \}$, let $\ket{\alpha}$ be defined as above, and let $d\geq 3$. \ Then there exist $d$ vectors $u_1 \ldots u_d$ in $\mathbb{CP}^{d-1}$ such that
\begin{itemize}

\item $\braket{u_i}{\phi_i} =0$ for all $i$.

\item $\braket{u_i}{\alpha}=\frac{1}{\sqrt{2}}$ for all $i$.

\item The tangent vectors $t_i$ to the geodesics from $u_i$ to $\alpha$ are linearly independent at $\alpha$ when the tangent space is viewed as a real vector space.

\end{itemize}
\end{lem}
\begin{proof}
Let $a=\sqrt{\frac{d}{2(d-1)^2}}$ and $b=\sqrt{\frac{d-2}{4(d-1)}}$. \ Then we define $u_1,...u_{d-2}$ as follows:
\ban u_i &=& \left(\sum_{j \neq i}a\ket{\phi_j}\right) + ib\ket{\phi_{i+1}} - ib\ket{\phi_{i+2}}\,, \ean
For the last two vectors, we set
\ban u_{d-1} &=& \left(\sum_{j \neq d-1}a\ket{\phi_j}\right) + ib\ket{\phi_d} - ib\ket{\phi_1}\,, \\
 u_d &=& \left( \sum_{j \neq d} a \ket{\phi_j}\right) + b \ket{\phi_1} - b \ket{\phi_2}\,. \ean
Note that the coefficients in $u_d$ are all real, unlike for the other $d-1$ vectors. \ It is straightforward to verify that $\braket{u_i}{\phi_i} =0$ and $\braket{u_i}{\alpha}=\frac{1}{\sqrt{2}}$ for all $i$. \ Furthermore, we can compute the tangent vectors $t_i$ as follows. \ The geodesics from $\ket{u_i}$ to $\ket{\alpha}$ in the Fubini-Study metric can be parameterized by
\[\gamma(t) = \cos(t)\ket{v_i} + \sin(t)\ket{\alpha}\]
where $v_i$ is the normalized component of $u_i$ orthogonal to $\alpha$, that is $v_i=k\left( u_i - \braket{\alpha}{u_i} \alpha\right)$ for some real normalization constant $k$. These geodesics lie entirely within our choice of gauge. \ Therefore $t_i$ is the projection of $\gamma'(t)|_{t=\pi/2}$ onto the plane orthogonal to $\alpha$, which is
\ban t_i = u_i - \braket{\alpha}{u_i} \alpha. \ean

Since $t_i$ is in the tangent space, its normalization is irrelevant. \ Also, since our gauge is fixed, there is no ambiguity involving the global phase of $u_i$ or $t_i$.

We now verify that the $t_i$'s are linearly independent. \ Suppose that $c_1 t_1 + \cdots + c_d t_d = 0$, with $c_i$ real. \ Note that $\braket{\alpha}{u_i} \alpha$ has all real coefficients, so a coefficient of $t_i$ is imaginary if and only if the corresponding coefficient of $u_i$ is imaginary. \ Since $c_1 t_1 + \cdots + c_d t_d = 0$, in particular the imaginary terms in $\ket{\phi_i}$ must sum to zero for all $i$. \ For $i=3\ldots d$, only the terms $c_{i-2}t_{i-2}$ and $c_{i-1} t_{i-1}$ contain imaginary multiples of $\ket{\phi_i}$. \ Hence this constraint implies $c_{i-2}=c_{i-1}$. \ Additionally, $c_1t_1$ is the only term containing an imaginary multiple of $\ket{\phi_2}$, so we must have $c_1=0$. \ Therefore $c_1=c_2=\cdots=c_{d-2}=0$. \ Since $c_{d-1}t_{d-1}$ is the only term containing an imaginary multiple of $\ket{\phi_1}$, we must have $c_{d-1}=0$, and hence $c_d=0$ as well. \ Therefore the $t_i$'s are linearly independent.
\end{proof}

Note that in a real Hilbert space, the analogous statement to Lemma \ref{ui_lemma} is false because the dimension of the tangent space at $\alpha$ is only $d-1$. \ In a complex Hilbert space the dimension of the tangent space is $2d-2$, so the tangent space can contain $d$ linearly independent vectors assuming $d\geq 2$.

We now show that Lemma \ref{ui_lemma} implies the existence of a set $B$ of positive measure, on which every $\mu_{u_i}$ is ``equivalent" to the Lebesgue measure, in the sense that if $S\subseteq B$ has positive Lebesgue measure, then $S$ has positive measure under each $\mu_{u_i}$.

\begin{lem}
\label{lem_product_measure}
Let $u_i$ and $\ket{\alpha}$ be as defined in Lemma \ref{ui_lemma}. \ Then there exists a set $B$ in the neighborhood of $\ket{\alpha}$, of positive Lebesgue measure, such that the $\mu_{u_i}$ are equivalent to the Lebesgue measure on $B$.
\end{lem}
\begin{proof}
Consider
\[
B= N_{\varepsilon}(\ket{\alpha}) \cap  \supp(\mu_{u_1}) \cap \supp(\mu_{u_2}) \cap \ldots \cap \supp(\mu_{u_d})
\]
where $N_{\varepsilon}(\ket{\alpha})$ denotes the $\varepsilon$-neighborhood of $\alpha$. \ For sufficiently small $\varepsilon$, one can show that $B$ has the desired properties using the same techniques as in the proof of Lemma \ref{rsupp}.
\end{proof}

From these two lemmas, the proof of our main theorem follows, since the orthogonality of each $u_i$ to $\phi_i$ (together with the Born rule) implies that the set $B$ cannot give any outcome with positive probability under measurement. \ But each element in $B$ must give \textit{some} outcome under measurement.

\begin{thm}
\label{main_thm}
There exists no symmetric, maximally nontrivial $\psi$-epistemic theory in
dimension $d\geq 3$.
\end{thm}
\begin{proof}
By Lemmas \ref{ui_lemma} and \ref{lem_product_measure}, there is a measurement basis $M=\{\phi_1,\ldots,\phi_d\}$ and vectors $u_1, \ldots, u_d$ such that each $u_i$ is orthogonal to $\phi_i$.

Furthermore, there is a set $B$ of positive measure such that each $\mu_{u_i}$ is equivalent to the Lebesgue measure on $B$. \ Therefore by the Born rule, for each $i$ we must have
\[
\int_B \mu_{u_i}(\lambda)\xi_{i,M}(\lambda) \mathrm{d}\lambda =0.
\]
Since each $\mu_{u_i}$ is equivalent to the Lebesgue measure on $B$, this implies
\[
\int_B \xi_{i,M}(\lambda) \mathrm{d}\lambda =0.
\]
But also, since $\sum_i \xi_{i,M}(\lambda) = 1$ for each state $\lambda$, we have that
\[
\displaystyle\sum_i \int_B \xi_{i,M}(\lambda) \mathrm{d}\lambda =  \int_B \mathrm{d}\lambda>0
\]
which is a contradiction. \qedhere

\end{proof}

\subsection{Extending the Proof to $\Lambda=U(d)$}

We now rule out a generalization of strongly symmetric, maximally nontrivial theories with a larger ontic space, namely $\Lambda=U(d)$.

Recall that a theory is strongly symmetric if $\Lambda=\mathbb{CP}^{d-1}$ and $\mu_{U\psi}(U\lambda)=\mu_{\psi}(\lambda)$ for all unitaries $U$. \ A theory is (weakly) symmetric if $\mu_{U\psi}(U\lambda)=\mu_{\psi}(\lambda)$ only for those $U$ such that $U\psi=\psi$.

We can generalize the definition of strong symmetry to any ontic space $\Lambda$ on which the unitary group has an action. \ We define a theory to be strongly symmetric with ontic space $\Lambda$ and action $a: U(d) \times \Lambda \rightarrow \Lambda$ if $\mu_{U\psi}(a(U,\lambda))=\mu_{\psi}(\lambda)$ for all unitaries $U$. \ In other words, for all $U$ and $\psi$, the following two distributions on $M$ are identical: draw $M$ from $\mu_{U\psi}$, or draw $N$ from $\mu_\psi$ and then set $M=a(U,N)$. \ Strongly symmetric theories admit a natural dynamics, since applying a unitary to a quantum state is equivalent to applying the unitary to the ontic states via the action $a$.

A natural choice of $\Lambda$ in this context is $U(d)$, the symmetry group of the $d$-dimensional Hilbert space. \ The unitary group $U(d)$ has a natural action on itself by left multiplication. \ We now show that there are no strongly symmetric, maximally nontrivial theories with $\Lambda=U(d)$ and with the action $a(U, \lambda)=U\lambda$ given by left multiplication.

Our proof will proceed similarly to the proof above. \ Assume there exists a strongly symmetric, maximally nontrivial theory with $\Lambda=U(d)$. \ We will find a set of quantum states $\{\Psi_1,\Psi_2,\Psi_3\}$, such that the $\mu_{\Psi_j}$ have nontrivial joint overlap, by which we mean that $\supp(\Psi_1)\cap \supp(\Psi_2) \cap \supp(\Psi_3) $ has positive measure. \ We will then create an orthonormal basis $\{e_1,\ldots,e_d\}$ such that for all $i$, there exists a $j$ such that $\braket{e_i}{\Psi_j}$=0. \ Therefore, if measured in the basis $\{e_i\}$, the ontic states in $S$ cannot give output $e_i$ for any $i$ by the Born rule, which contradicts the fact that they must give some outcome under measurement.

To show this, we first need to characterize strongly symmetric theories with $\Lambda=U(d)$. \ We now show that the probability distributions $\mu_\psi$ of any strongly symmetric theory are fully characterized by some probability measure $\nu$ on $\mathbb{CP}^{d-1}$.

\begin{lem} Let $\pi=(\Lambda,\mu,\xi)$ be a strongly symmetric theory with $\Lambda=U(d)$ and with action $a(U, \lambda)=U\lambda$ given by left multiplication. \ Then there exists a probability measure $\nu$ on $\mathbb{CP}^{d-1}$ which fully characterizes the probability distributions $\mu_\psi$ for all $\psi$. \ In particular, to draw a sample $M$ from $\mu_\psi$, one can first draw $\lambda \in \mathbb{CP}^{d-1}$ from $\nu$, and then draw $M$ uniformly (according to the Haar measure) such that $M^\dagger \psi = \lambda$.
\label{lem:existence-nu}
\end{lem}

\begin{proof}
Suppose we draw $M$ from $\mu_\psi$. \ Let $\nu_\psi$ be the distribution on $\mathbb{CP}^{d-1}$ induced by $M^\dagger \psi$.

Suppose that $U\psi=\psi$ for some unitary $U$. \ Then by symmetry, $\mu_\psi$ must be invariant under applying $U$. \ Note that if $M^\dagger \psi= \lambda$, then $(UM)^\dagger \psi = M^\dagger U^\dagger \psi = M^\dagger \psi = \lambda$ as well.

Let $\mu_{\psi, \lambda}$ be the measure over $M$'s obtained by starting from $\mu_\psi$ and then conditioning on $M^\dagger \psi= \lambda$. \ By the above observation, if $\mu_\psi$ is invariant under every such $U$, then $\mu_{\psi, \lambda}$ must also be invariant under $U$ for every $\lambda$. \ In particular this implies that $\mu_{\psi, \lambda}$ must be the uniform (Haar) measure on matrices $M$ such that $M^\dagger \psi = \lambda$. \ Therefore, to draw a sample $M$ from $\mu_\psi$, one can first draw $\lambda$ from $\nu_\psi$, and then draw $M$ uniformly (according to the Haar measure) such that $M^\dagger \psi = \lambda$.

Now suppose that $U\psi=\phi$. \ Let $M$ be drawn from $\mu_\phi$ and $N$ be drawn from $\mu_\psi$. By strong symmetry, the distribution of $UN$ is the same as the distribution of $M$. \ But we also know that $(UN)^\dagger \phi = N^\dagger U^\dagger \phi = N^\dagger \psi$. \ Hence the induced distribution of $(UN)^\dagger \phi$ is the same as the induced distribution of $N^\dagger \psi$. But by strong symmetry the former distribution is $\nu_\phi$, and the latter distribution is $\nu_\psi$. Hence $\nu_\psi=\nu_\phi = \nu$ for all $\phi$ and $\psi$ as desired.

\end{proof}

By the Lebesgue decomposition theorem, $\nu$ can be uniquely decomposed as $\nu = \nu_S + \nu_C$, where $\nu_C$ is absolutely continuous with respect to the Lebesgue measure on $\mathbb{CP}^{d-1}$, and $\nu_S$ is singular with respect to the Lebesgue measure on $\mathbb{CP}^{d-1}$. \ As in the previous section, we now show that maximum nontriviality implies that $\nu_C$ has positive total measure, and we restrict our attention to $\nu_C$ in future parts.

\begin{lem}
\label{lem:nu-C-nontrivial}
In any maximally nontrivial, strongly symmetric theory with $\Lambda=U(d)$ and action $a$ given by left multiplication, $\nu_C$ has positive total measure.

\end{lem}
\begin{proof}

By Lemma \ref{lem:existence-nu}, if $M$ is drawn from $\mu_\psi$, then $M^\dagger$ maps $\psi$ to a state $\lambda$ chosen from $\nu$, and maps $\phi$ uniformly at random to a state $\lambda'$ with inner product $|\braket{\psi}{\phi}|=|\braket{\lambda}{\lambda'}|$. \ Likewise, if $N$ is drawn from $\mu_\phi$, then $N^\dagger$ maps $\phi$ to a state $\lambda'$ chosen from $\nu$, and maps $\psi$ uniformly at random to a state $\lambda$ with inner product $|\braket{\psi}{\phi}|=|\braket{\lambda}{\lambda'}|$.

By maximum nontriviality, $\mu_\psi$ and $\mu_\phi$ have nontrivial overlap for any non-orthogonal $\psi$ and $\phi$. \ Let $S=\supp(\nu)$. \ Since $\mu_\psi$ and $\mu_\phi$ have nontrivial overlap for all non-orthogonal $\psi$ and $\phi$, we must have that for all $r\in(0,1]$, if $\lambda$ is chosen according to $\nu$ and $\lambda'$ is chosen uniformly such that $|\braket{\lambda}{\lambda'}|=r$, then $\lambda' \in S$ with positive probability. \ In particular there must exist some $\lambda \in S$ such that, if $r$ is chosen uniformly at random in $(0,1)$, and $\lambda'$ is chosen uniformly at random such that $|\braket{\lambda}{\lambda'}|=r$, then $\lambda' \in S$ with positive probability. \ This immediately implies that $S$ has positive Lebesgue measure, since the Lebesgue measure of $S$ can be expressed as the integral of the indicator for $\lambda' \in S$ over $\mathbb{CP}^{d-1}$, which can be written in polar coordinates centered about $\lambda$. \ The integral is positive by the preceding observation. \ Hence $S=\supp(\nu)$ has positive Lebesgue measure and $\nu_C$ has positive total measure.

\end{proof}

As before, we now assume that $\nu =\nu_{C}$, i.e.\ we will discard the singular part of $\nu$. \ We can do this without loss of generality since our contradiction will not rely on the normalization of $\nu$.

Next we will show that it is easy to find states $\{\Psi_1,\Psi_2,\Psi_3\}$ such that the $\mu_{\Psi_i}$ have nontrivial joint overlap, i.e.\ the intersection of their supports is a set of positive measure. \ The proof will make crucial use of the Lebesgue density theorem. \ The Lebesgue density theorem says that for any set $S$ of positive measure, for almost every point in $S$, the density of $S$ at that point is $1$. \ More formally, the density of $S$ at point $x$, denoted $d_x(S)$, is defined as
\[
d_x(S) = \lim_{\epsilon \rightarrow 0^+} \frac{\mu(B_{\epsilon}(x) \cap S)}{\mu(B_{\epsilon}(x))},
\]
where $\mu$ is the Lebesgue measure and $B_{\epsilon}(x)$ is the $\epsilon$-ball centered at $x$. \ The Lebesgue density theorem says that for any measurable set $S$, the set $T=\{x\in S : d_x(S)=1\}\subseteq S$ differs from $S$ by at most a set of measure zero. \ In particular $T$ has the same measure as $S$. \ The points $x\in S$ such that $d_x(S)=1$ are called the Lebesgue density points of $S$.

\begin{lem}
\label{lem:joint-overlap-properties}

For any strongly symmetric theory with $\Lambda=U(d)$, there exists a set $T\subseteq \mathbb{CP}^{d-1}$ such that:

\begin{enumerate}
\item $T$ has positive measure. \label{Tposmeas}
\item If $M^\dagger \psi \in T$ then $M$ is a Lebesgue density point of $\supp(\mu_\psi)$. \label{Tldp}
\item If $\mu_\psi$ and $\mu_\phi$ have nontrivial overlap, then there exists an $M$ such that $M^\dagger \psi \in T$ and $M^\dagger \phi \in T$. \label{Tintersection}
\item For any three states $\Psi_1,\Psi_2,\Psi_3$ such that there exists an $M$ with $M^\dagger \Psi_i \in T$ for each $i=1,2,3$, the intersection $\supp(\mu_{\Psi_1}) \cap \supp(\mu_{\Psi_2}) \cap \supp(\mu_{\Psi_3})$ has positive measure. \label{Tjointint}
\end{enumerate}

\end{lem}

\begin{proof}

Let $S=\supp(\nu)$ and let $S'_\psi = \supp(\mu_\psi)$. \ Note that $S$ and $S'_\psi$ have positive measure by Lemma \ref{lem:nu-C-nontrivial}. \ Let $T'_\psi$ be the set of Lebesgue density points of $S_\psi$. \ Let $T=\{M^\dagger \psi : M \in T'_\psi\}$.

First note that the definition of $T$ is independent of our choice of $\psi$. \ Indeed by strong symmetry, if $U\psi=\phi$ then $U T'_\psi = T'_\phi$, where the notation $U S$ denotes the set $\{Us : s\in S\}$. \ Therefore
\[
T=\{M^\dagger \psi : M \in T'_\psi\} = \{ (UM)^\dagger U \psi : M \in T'_\psi\} = \{M^\dagger \phi : M \in T'_\phi\}.
\]

Second, note that $T\subseteq S$ since $T'_\psi \subseteq S'_\psi$. \ Furthermore, since $T'_\psi$ differs from $S'_\psi$ only by a set of measure zero, $T$ differs from $S$ only by a set of measure zero as well. \ Hence $T$ has positive measure, which proves property \ref{Tposmeas}.

Also, $T$ precisely characterizes $T'_\psi$, in the sense that $M\in T'_\psi$ if and only if $M^\dagger \psi \in T$. \ The ``if" direction follows directly from the definition, while the ``only if" direction follows from the fact that once $M^\dagger \psi$ is fixed, the distribution of $M$ under $\mu_\psi$ is uniform over the remaining degrees of freedom by Lemma \ref{lem:existence-nu}. \ This proves property \ref{Tldp}.

Next suppose that $\mu_\psi$ and $\mu_\phi$ have nontrivial overlap, i.e.\ $\supp(\mu_\psi) \cap \supp(\mu_\phi)$ has positive measure. \ Then by Lemma \ref{lem:existence-nu}, the set $\{M: M^\dagger \psi \in S \text{ and } M^\dagger \phi \in S\}$ has positive measure as well. \ Since $T$ differs from $S$ only by a set of measure zero, this implies that $\{M: M^\dagger \psi \in T \text{ and } M^\dagger \phi \in T\}$ has positive measure also. \ Hence there exists an $M$ such that $M^\dagger \psi \in T$ and $M^\dagger \phi \in T$. \ This proves property \ref{Tintersection}.

Finally suppose that for three states $\Psi_1,\Psi_2,\Psi_3$, there exists an $M$ with $M^\dagger \Psi_i \in T$ for each $i=1,2,3$. \ By property \ref{Tldp} we know that $M$ is a Lebesgue density point of $\supp(\mu_{\Psi_i})$ for each $i$. \ Suppose that we perturb $M$ by a small amount $\epsilon >0$ uniformly at random to obtain a new matrix $N$. \ More formally, choose $N$ according to the Haar measure on $B_\epsilon (M)$. \ Since $M$ is a Lebesgue density point of $\supp(\mu_{\Psi_i})$, by choosing small enough $\epsilon$, the density of $\supp(\mu_{\Psi_i})$ in $B_\epsilon (M)$ can be made arbitrarily close to $1$. \ Therefore, by choosing $\epsilon$ small enough, we can ensure that for each $i$, the probability that $N^\dagger \Psi_i \in T$ is at least (say) $0.99$.

The events $N^\dagger \Psi_i \in T$ are not necessarily independent, but by the union bound the probability that $N^\dagger \Psi_i \in T$ for all $i=1,2,3$ is at least $0.97$. \ Hence $N$ will be in $\supp(\mu_{\Psi_1}) \cap \supp(\mu_{\Psi_2}) \cap \supp(\mu_{\Psi_3})$ with positive probability. \ This implies that $ \supp(\mu_{\Psi_1}) \cap \supp(\mu_{\Psi_2}) \cap \supp(\mu_{\Psi_3})$ has positive Lebesgue measure, and hence the $\mu_{\Psi_i}$ have nontrivial joint overlap, which proves property \ref{Tjointint}.

\end{proof}

Now we will show that if $\Psi_1$, $\Psi_2$ and $\Psi_3$ are chosen appropriately, then there exists an orthonormal basis $\{e_1,\ldots,e_d\}$ such that for all $i$, there exists a $j$ such that $\braket{e_i}{\psi_j}$=0. \ In particular such a basis exists if $\Psi_1$ and $\Psi_2$ are ``nearly orthogonal" and $\Psi_3$ is not coplanar with $\Psi_1$ and $\Psi_2$.

\begin{lem}
\label{lem:ortho-basis-psii}

Let $u_1$ and $u_2$ be orthonormal vectors, and let $\Psi_3$ be a state that is not coplanar with $u_1$ and $u_2$. \ In particular assume that $|\braket{\Psi_3}{u_1}|^2+|\braket{\Psi_3}{u_2}|^2\leq k$ for some fixed $k<1$. \ Then there exists a $k'>0$  (depending on $k$) such that if $\Psi_1=u_1$, and $\Psi_2$ is in the $u_1,u_2$ plane such that $0<|\braket{\Psi_1}{\Psi_2}|<k'$, then there exists an orthonormal basis $\{e_1, \ldots, e_d\}$ such that for all $i=1\ldots d$, there exists $j \in \{1,2,3\}$ such that $\braket{e_i}{\psi_j}=0$.

\end{lem}
\begin{proof}

Set $\Psi_1=u_1$. \ Without loss of generality we can set
\[
\Psi_2 = \frac{a u_1 + u_2}{\sqrt{|a|^2 + 1}}
\] for some complex parameter $a$ which we have yet to specify. \ By the Gram-Schmidt process, there exists a vector $u_3$, and complex coefficients $b,c$, such that
\[
\Psi_3 = \frac{b u_1 + c u_2 + u_3}{\sqrt{|b|^2 + |c|^2 + 1}}
\]
\noindent where $\{u_1, u_2, u_3\}$ is an orthonormal basis for the subspace spanned by $\Psi_1$, $\Psi_2$ and $\Psi_3$. \ Note that the statement $|\braket{\Psi_3}{u_1}|^2+|\braket{\Psi_3}{u_2}|^2\leq k$ implies that $b, c \leq f(k)$ for some function $f$ of $k$.

Now consider the following three (non-normalized) vectors, parameterized by $x\in \mathbb{C}$:
\begin{align*}
e_1 &=      x u_2 +           u_3 \\
e_2 &= u_1 -  a^* u_2 + a^* x^* u_3 \\
e_3 &= a (1 + x^* x) u_1 +      u_2 -       x^* u_3
\end{align*}

By construction the $e_i$'s are orthogonal to one another, $e_1$ is orthogonal to $\Psi_1$ and $e_2$ is orthogonal to $\Psi_2$. \ We would like to have
\[\braket{e_3}{\Psi_3}=ba^* (1 + |x|^2) + c - x = 0\]
as well. \ If either $a=0$ or $b=0$, we can achieve this by simply setting $x = c$, so assume $a$ and $b$ are nonzero. \ Also, let $w = \frac{ba^*}{|ba^*|}$, which has norm $1$. \ Then setting $\braket{e_3}{\Psi_3}=0$ is equivalent to setting
\[
|ba^*| (1 + |x|^2) + cw - xw = 0
\]

\noindent Setting $xw = p + iq$ for real $p$,$q$ gives
\begin{align*}
|ba^*| (1 + p^2 + q^2) + \mathrm{Re}(cw) - p &= 0\\
\mathrm{Im}(cw) - q &= 0.
\end{align*}

\noindent Plugging the value of $q$ from the second equation into the first gives a quadratic in $p$,
\[
 |ba^*| p^2 - p + |ba^*| \left(1 + \mathrm{Im}(cw)^2\right) + \mathrm{Re}(cw) = 0.
\]
This has a real solution in $p$ if and only if
\[
1 - 4 |ba^*| \left(|ba^*|\left(1+ \mathrm{Im}(cw)^2\right) + \mathrm{Re}(cw)\right) \geq 0.
\]
Note that $b$,$c$ are bounded above by $f(k)$ as noted above. \ By making $\Psi_1$ nearly orthogonal to $\Psi_2$, we can choose $a$ arbitrarily close to zero. \ This makes the left-hand side of the inequality arbitrarily close to $1$, so the inequality will hold. \ Hence, we can solve for $x$ and find an $e_3$ that is orthogonal to $\Psi_3$. \ Therefore if $|\braket{\Psi_1}{\Psi_2}|<k'$ for some $k'$ which depends on $k$, then the desired orthonormal basis exists.

This gives a basis $e_1$, $e_2$, $e_3$ for the subspace spanned by $\Psi_1$, $\Psi_2$ and $\Psi_3$ such that $\braket{e_i}{\Psi_i}=0$ for each $i=1,2,3$. \ If the dimension of the space $d$ is more than $3$, any extension of this basis to $e_1 \ldots e_d$ has the property that for all $i=1\ldots d$, there exists a $j \in \{1,2,3\}$ such that $\braket{e_i}{\psi_j}=0$.

\end{proof}

Our no-go theorem follows from the above observations.

\begin{thm}
\label{main_thm_unitary}
There are no strongly symmetric, maximally nontrivial $\psi$-epistemic theories with $\Lambda= U(d)$ in dimension $d\geq 3$.
\end{thm}
\begin{proof}

Suppose a strongly symmetric, maximally nontrivial theory exists with $\Lambda= U(d)$ in dimension $d\geq 3$. \ We will find three states $\Psi_1$, $\Psi_2$, and $\Psi_3$ with the following two properties:

\begin{enumerate}
\item[(A)] $\mu_{\Psi_1}$, $\mu_{\Psi_2}$ and $\mu_{\Psi_3}$ have nontrivial joint overlap, and
\item[(B)] there exists an orthonormal basis $e_1, \ldots, e_d$ such that for all $i=1,\ldots, d$, there exists a $j \in \{1,2,3\}$ such that $\braket{e_i}{\psi_j}=0$.
\end{enumerate}

The contradiction will follow as in the proof of Theorem \ref{main_thm}. \ Consider making a measurement in the basis $e_1, \ldots, e_d$. \ Let $S=\supp(\mu_{\Psi_1}) \cap \supp(\mu_{\Psi_2}) \cap \supp(\mu_{\Psi_3})$, and consider ontic states $M\in S$. \ Note that $M\in \supp(\mu_{\Psi_j})$, and that for each $j=1,2,3$ and each $e_i$, there exists a $j \in \{1,2,3\}$ such that $\braket{e_i}{\Psi_j}=0$. \ So by the Born rule, for each $i$ at most a set of measure zero of $M\in S$ can return answer $e_i$ with positive probability. \ Since each $M\in S$ gives some outcome $e_i$ under measurement, $S$ must have measure zero. \ But $S$ has positive measure by property (A), which is a contradiction.

Using Lemma \ref{lem:joint-overlap-properties}, one can show that for all nonorthogonal states $\Psi_1$ and $\Psi_2$, there exists a state $\Psi_3$ not coplanar with $\Psi_1$ and $\Psi_2$ such that the states have property (A).\footnote{Indeed, by Lemma \ref{lem:joint-overlap-properties} there exists a set $T$ with the four properties defined in the lemma. Given $\Psi_1$ and $\Psi_2$ which are non-orthogonal, by property \ref{Tintersection} of $T$ there exists an $M$ such that $M^\dagger \Psi_1 \in T$ and $M^\dagger \Psi_2 \in T$. \ Since $T$ has positive measure by property \ref{Tposmeas}, there exists a $\Psi_3$ not coplanar with $\Psi_1$ and $\Psi_2$ such that $M^\dagger \Psi_3 \in T$ as well. \ These three states have nontrivial joint overlap by property \ref{Tjointint} of $T$.} \ By Lemma \ref{lem:ortho-basis-psii}, if $\Psi_1$ and $\Psi_2$ are ``nearly orthogonal," and $\Psi_3$ is not coplanar with $\Psi_1$ and $\Psi_2$, then the states have property (B). \ These two facts nearly suffice to guarantee the existence of three states with properties (A) and (B), but they fall short. \ The reason is that the degree to which $\Psi_1$ and $\Psi_2$ must be ``nearly orthogonal" depends on the choice of $\Psi_3$. \ Although for every $\Psi_1$ and $\Psi_2$ there \emph{exists} a non-coplanar $\Psi_3$ which has property (A), the choice of $\Psi_3$ could depend arbitrarily on $\Psi_1$ and $\Psi_2$, so in particular could be such that property (B) is not satisfied.

To fix this, we will consider an infinite sequence of quantum states $\psi_n$ and $\phi_n$ of decreasing inner product. \ For each $n$, there exists a $\chi_n$ that shares property (A) with $\psi_n$ and $\phi_n$. \ Since $\mathbb{CP}^{d-1}$ is compact, even though the $\chi_n$'s might ``wander" in $\mathbb{CP}^{d-1}$, by passing to a subsequence we can assume they converge to fixed state $\chi$, which we will show can be chosen to be non-coplanar with $\psi_n$ and $\phi_n$ for large $n$. \ This allows us to ``nail down" the choice of $\chi_n$ so that we can apply Lemma \ref{lem:ortho-basis-psii} to $\Psi_1=\psi_n$, $\Psi_2=\phi_n$, and $\Psi_3=\chi_n$.

More precisely, by Lemma \ref{lem:joint-overlap-properties}, there exists a set $T$ with the following properties:
\begin{enumerate}
\item $T$ has positive measure.
\item If $M^\dagger \psi \in T$ then $M$ is a Lebesgue density point of $\supp(\mu_\psi)$.
\item If $\mu_\psi$ and $\mu_\phi$ have nontrivial overlap, then there exists an $M$ such that $M^\dagger \psi \in T$ and $M^\dagger \phi \in T$.
\item For any three states $\Psi_1,\Psi_2,\Psi_3$ such that there exists an $M$ with $M^\dagger \Psi_i \in T$ for each $i=1,2,3$, the intersection $\supp(\mu_{\Psi_1}) \cap \supp(\mu_{\Psi_2}) \cap \supp(\mu_{\Psi_3})$ has positive measure.
\end{enumerate}
Consider a sequence of quantum states $\psi_n$, $\phi_n$ such that $\psi_n \rightarrow \psi$ and $\phi_n \rightarrow \phi$ for orthogonal states $\psi$ and $\phi$ as $n\rightarrow \infty$, but $|\braket{\psi_n}{\phi_n}|>0$ for all $n\in\mathbb{N}$. \ For each $n$, by maximum nontriviality $\mu_{\psi_n}$ and $\mu_{\phi_n}$ have nontrivial overlap, and there exists an $M_n$ such that $M_n^\dagger \psi_n \in T$ and $M_n\dagger \phi_n \in T$ by property \ref{Tintersection} of $T$.

Let $\tilde{\psi}_n = M_n^\dagger \psi_n$, and let $\tilde{\phi}_n =M_n^\dagger \phi_n$. \ By construction $\tilde{\psi}_n, \tilde{\phi}_n \in T$. \ These form sequences in $\mathbb{CP}^{d-1}$. \ Since $\mathbb{CP}^{d-1}$ is compact, there exists a subsequence of the $\tilde{\psi}_n$'s which approaches some $\tilde{\psi}$ as $n\rightarrow \infty$. \ Therefore by passing to a subsequence, there exists $\tilde{\psi}$ and $\tilde{\phi}$ such that $\tilde{\psi}_n \rightarrow \tilde{\psi}$ and $\tilde{\phi}_n \rightarrow \tilde{\phi}$ as $n\rightarrow \infty$.

Since $T$ has positive measure by property \ref{Tposmeas}, there exists a $\tilde{\chi} \in T$ that is not coplanar with $\tilde{\psi}$ and $\tilde{\phi}$. \ Fix such a $\tilde{\chi}$, and let $\chi_n = M_n \tilde{\chi}$. \ Passing to a subsequence again, we have that $\chi_n \rightarrow \chi$ for some state $\chi$. \ Note $\chi$ is not coplanar with $\psi$ and $\phi$, since the $M_n^\dagger$'s preserve inner products. \ Also note that for each $n$, we have $M_n^\dagger \psi_n, M_n^\dagger \phi_n, M_n^\dagger \chi_n \in T$. \ So by property \ref{Tjointint} of $T$, the measures $\mu_{\psi_n}$, $\mu_{\phi_n}$ and $\mu_{\chi_n}$ have property (A) (nontrivial joint overlap) for each $n$.

Now as $n\rightarrow \infty$, we have that $|\braket{\psi_n}{\phi_n}| \rightarrow 0$, and yet at the same time $\chi_n \rightarrow \chi$ for some fixed state $\chi$ which is not coplanar with $\psi$ and $\phi$. \ Hence for sufficiently large $n$, if $e_{1_n}$ and $e_{2_n}$ span the $\psi_n,\phi_n$ plane, then we will have $|\braket{e_{1_n}}{\chi_n}|^2+|\braket{e_{2_n}}{\chi_n}|^2\leq k$ for some $k<1$. \ By Lemma \ref{lem:ortho-basis-psii}, there exists a $k'$ such that if $|\braket{\psi_n}{\phi_n}|<k'$, then property (B) holds for $\psi_n$, $\phi_n$ and $\chi_n$. \ For sufficiently large $n$, we have $|\braket{\psi_n}{\phi_n}|<k'$, since $|\braket{\psi_n}{\phi_n}|\rightarrow 0$, and hence property (B) holds for these states.

Putting this together, we can find a value of $n$ such that the three states $\Psi_1=\psi_n$, $\Psi_2=\phi_n$, and $\Psi_3=\chi_n$ have both properties (A) and (B). \ The contradiction follows as noted above.

\end{proof}

\section{Conclusions and Open Problems}
\label{sec-conclusion}
In this paper, we gave a construction of a maximally nontrivial theory
in arbitrary finite dimensions. \ However, the theory we constructed is not
symmetric and is rather unnatural. \ We then proved that symmetric, maximally
nontrivial $\psi$-epistemic theories do not exist in dimensions $d\geq 3$ (in contrast to the $d=2$ case,
where the Kochen-Specker theory provides an example). \ Our impossibility
proof made heavy use of the symmetry assumption. \ As for the assumption $d\geq 3$, we
used that in two places: firstly and most importantly, to get a nonempty deficiency region (in Corollary \ref{deficient}), and secondly, to prove that $r_{\psi}=\frac{1}{2}$ in Lemma \ref{rhalf}.

It might be possible to relax our symmetry assumption and obtain no-go theorems for different ontic spaces,
since deficiency holds for any $\psi$-epistemic theory in $d\geq 3$ even
without a symmetry assumption. \ As shown above, we can generalize our no-go theorem to rule out strongly symmetric maximally nontrivial theories with $\Lambda=U(d)$. \ It would be particularly interesting to know whether (merely) symmetric theories exist in $\Lambda=U(d)$, or in any ontic spaces $\Lambda$ that are larger than $\mathbb{CP}^{d-1}$, but that are still acted on by the $d$-dimensional unitary group.

Also, in our proof, we did not use the
specific form of the Born rule, only the fact that projection of $\ket{\psi}$ onto $\ket{\phi}$ must occur with
probability $0$ if $\braket{\psi}{\phi}=0$. \ Additional properties of the Born rule might place further constraints on
$\psi$-epistemic theories.

Interestingly, trying to generalize the proof of Theorem \ref{strong-symmetry-thm} directly to obtain a proof of Theorem \ref{main_thm} gives rise to a  variant of the Kakeya/Besicovitch problem. \ Recall that to prove Theorem \ref{strong-symmetry-thm}, we showed that ontic states in a set $B$ in the neighborhood of $\alpha$ returned value $j$ under measurement, and yet the average measure of states orthogonal to $j$ had nontrivial support on $B$. \ Now if the measures $\mu_\psi = f_\psi (|\braket{\psi}{\lambda}|^2)$ vary with $\psi$, it remains open whether or not the measures of states orthogonal to $j$ must have support on $B$, or if instead it is possible for them to ``evade" $B$ to avoid contradicting the Born rule.

Placing this problem in the plane rather than in $\mathbb{CP}^{d-1}$, we obtain a clean Kakeya-like problem as follows. \ Let $S$ be a subset of $\mathbb{R}^2$ with the following property. \ For all $x \in \mathbb{R}^2$ and $\varepsilon > 0$,  $S$ contains a set of circles, centered at $x$, that has positive Lebesgue measure within the annulus $ \{ y : | y-x | \in [1-\varepsilon,1] \}$. \ Can the complement of $S$ have positive Lebesgue measure? \ This question has been discussed on MathOverflow \cite{mathoverflow2} but remains open.

Here are some additional open problems.

\begin{itemize}
\item An obvious problem is
whether symmetric and nontrivial (but not necessarily \emph{maximally} nontrivial) theories exist in dimensions $d\geq 3$.

\item How does the size of the deficiency region scale as the dimension $d$ increases?

\item In the maximally nontrivial theory we constructed, the overlap between any two non-orthogonal
states $\ket{\psi},\ket{\phi}$ is vanishingly small: like $(\varepsilon / d)^{O(d)}$ as a function of the dimension $d$ and inner product $\varepsilon = \left| \braket{\psi}{\phi} \right|$. \ Is it possible to construct a theory with substantially higher overlaps -- say, $(\varepsilon / d)^{O(1)}$? \ (Note that if $d\geq 3$, then the result of Leifer and Maroney \cite{marleif} says that the overlap cannot achieve its ``maximum'' value of $\varepsilon^2$.)

\item Can we construct $\psi$-epistemic theories with the property that an ontic state $\lambda$, in the support of an ontic distribution $\mu_{\psi}$, can \emph{never} be used to recover the quantum state $\psi$ uniquely? \ (This question was previously asked by Leifer and Maroney \cite{marleif}, as well as by A.\ Montina on MathOverflow \cite{mathoverflow}.)

\item What can be said about the case of infinite-dimensional Hilbert spaces?

\end{itemize}

\section{Acknowledgments}

We thank Andrew Drucker for discussions about the measure-theoretic aspects of this paper, and Matt Leifer and Terry Rudolph for discussions about $\psi$-epistemic theories in general. \ We also thank Terence Tao, Sean Eberhard, Ramiro de la Vega, Scott Carnahan, and Alberto Montina for discussions on MathOverflow.

\end{document}